\documentclass[letterpaper, 10 pt, conference]{ieeeconf}

\IEEEoverridecommandlockouts
\title{\LARGE \bf Distributed Consistent Data Association
}

\author{Spyridon Leonardos$^{1}$,  Xiaowei Zhou$^{2}$ and Kostas Daniilidis$^{3}$% <-this % stops a space
\thanks{$^{1,2,3}$The authors are with the  Department of Computer and Information Science, University of Pennsylvania
        Philadelphia, PA 19104, USA
        {\tt\small \{spyridon,xiaowz,kostas\}@seas.upenn.edu}}
        \thanks{Support by the grants ARL MAST-CTA W911NF-08-2-0004 and ARL RCTA W911NF-10-2-0016 is gratefully acknowledged.}
}

\usepackage{amsmath,amsfonts,amssymb}
\usepackage{bm}
\usepackage{verbatim}
\usepackage{comment}
\usepackage{graphicx}
\usepackage{xcolor}
\usepackage{tikz}
\usetikzlibrary{arrows}
\usepackage{pgfplots}
\usepackage{cite}
\usepackage{hyperref}
\newcommand*{\QEDA}{\hfill\ensuremath{\blacksquare}}%
\newcommand{\todo}[1]{\textcolor{red}{TODO: #1}}
\newcommand{\notion}[1]{\textit{#1}}

\newtheorem{thm}{Theorem}[section]

\newtheorem{definition}[thm]{Definition}

\newtheorem{lemma}[thm]{Lemma}

\newtheorem{remark}{Remark}

\DeclareMathOperator*{\rank}{rank}

\DeclareMathOperator*{\diag}{diag}

\DeclareMathOperator*{\trace}{tr}

\newcommand{\mean}[1]{\overline{#1}}

\newcommand{\cD}{\mathcal{D}}

\newcommand{\cNi}{\mathcal{N}_i}

\newcommand{\mapdef}[3]{#1 : #2 \rightarrow #3}

\newcommand{\reals}{\mathbb{R}}

\newcommand{\GL}{\mathbf{GL}}

\newcommand{\OO}[1]{\mathbf{O}({#1})}

\newcommand{\symgroup}[1]{\mathfrak{S}_{#1}}

\newcommand{\jiE}{{(j,i) \in \mathcal{E}}}
\newcommand{\ijEu}{{\{i,j\} \in \mathcal{E}}}

\newcommand{\card}[1]{\left| #1 \right|} % cardinality of a set
\newcommand{\graph}{\mathcal{G}}

\newcommand{\vset}{\mathcal{V}}

\newcommand{\eset}{\mathcal{E}}

\newcommand{\degmaxG}{d_{\textrm{max}}(\graph)}
\newcommand{\degmax}{d_{\textrm{max}}}

\newcommand{\degavg}{d_{\textrm{avg}}}

\newcommand{\Pii}{\Pi_i}
\newcommand{\tpi}{\widetilde{\pi}}
\newcommand{\Pij}{\Pi_j}

\newcommand{\tPi}{\widetilde{\Pi}}

\newcommand{\tPiij}{\widetilde{\Pi}_{ij}}
\newcommand{\tPiik}{\widetilde{\Pi}_{ik}}

\newcommand{\tPijk}{\widetilde{\Pi}_{jk}}

\newcommand{\ie}{i.e.}

\newcommand{\norm}[1]{\lVert#1\rVert}

\newcommand{\metric}[2]{\langle #1, #2\rangle}

\usepackage{algorithm}
\usepackage{algpseudocode}
\algnewcommand\algorithmicinput{\textbf{INPUT:}}

\definecolor{myblue}{RGB}{80,80,160}
\definecolor{mygreen}{RGB}{80,160,80}
\definecolor{myred}{RGB}{160,80,80}
\usetikzlibrary{positioning,chains,fit,shapes,calc,arrows,matrix}

\begin{document}

\maketitle
\thispagestyle{empty}
\pagestyle{empty}

%%%%%%%%%%%%%%%%%%%%%%%%%%%%%%%%%%%%%%%%%%%%%%%%%%%%%%%%%%%%%%%%%%%%%%%%%%%%%%%%
\begin{abstract}

Data association is one of the fundamental problems in multi-sensor systems.
Most current techniques rely on pairwise data associations
which can be spurious even after the employment of outlier rejection schemes.
Considering multiple pairwise associations at once  significantly increases accuracy and leads to consistency.
In this work, we propose two fully decentralized methods for consistent global data association from pairwise data associations.
The first method is a consensus algorithm on the set of doubly stochastic matrices.  The second method is a decentralization of the  spectral method proposed  by Pachauri et al. \cite{pachauri2013solving}.
We demonstrate the effectiveness of both methods using theoretical analysis and experimental evaluation.
%Finally, we prove that even deciding whether a consistent global data association exists is NP-complete.
\end{abstract}

%%%%%%%%%%%%%%%%%%%%%%%%%%%%%%%%%%%%%%%%%%%%%%%%%%%%%%%%%%%%%%%%%%%%%%%%%%%%%%%%
\section{INTRODUCTION} % and related work

Multi-sensor data association has been a long standing problem in robotics and computer vision. It refers to the problem of establishing correspondences between feature points, regions, or objects observed by different sensors and serves as the basis for many high-level tasks such as localization, mapping and planning. Most of the efforts in previous works have been dedicated to improving the data association by designing new feature detectors, descriptors, and outlier removal algorithms in a pairwise setting. However, the problem setting in practice is often multi-way if the data are collected by a sensor network or a multi-robot system, and how to establish consistent data association for multiple sensors and leverage the global reasoning to improve the association has drawn increasing attention in recent years.

A necessary condition for good data association among multiple sensors is the cycle consistency meaning that the composition of associations along a cycle of sensors should be identity. The cycle consistency is often violated in practice due to outliers and how to use such a constraint to remove the false associations has been studied in robotics, computer vision and graphics. The related work will be discussed in Section \ref{sec:related}. While promising empirical and theoretical results have been reported in previous works, most of them addressed the problem in a centralized manner and assumed that the observations and states of all sensors could be accessed at the same time. This assumption is impractical in a distributed system where the data is processed on local sensors with limited computational and communication resources. In this paper, we aim to develop distributed algorithms that can efficiently operate on each sensor,  requiring only local information and communication with its neighbors, and finally producing globally consistent data association.

 \begin{figure}
 \label{firstfig}
 \centering
\includegraphics[trim={5cm 4cm 5cm 2cm},clip,scale=0.3]{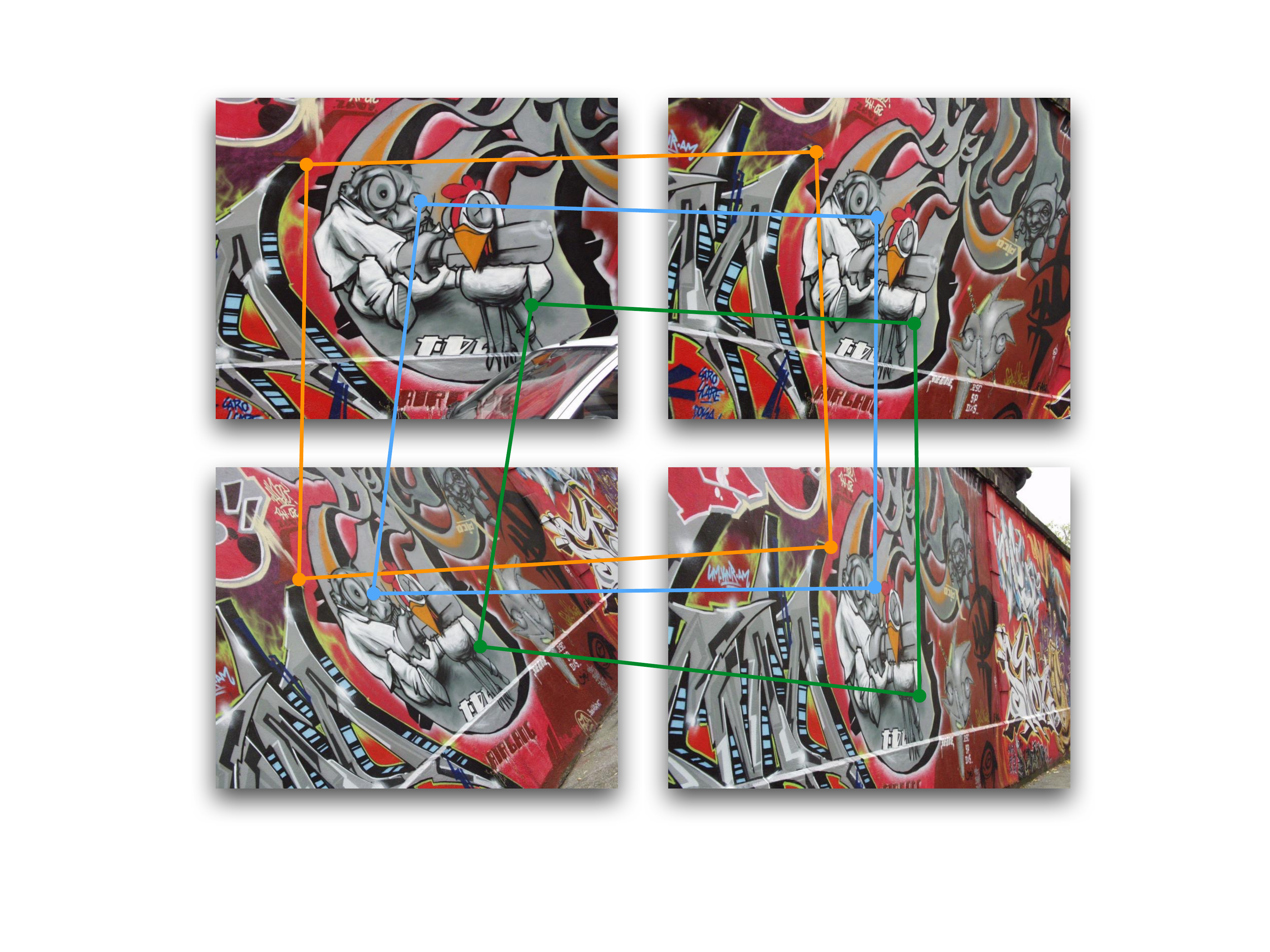}
%  trim={<left> <lower> <right> <upper>}
\caption{Illustration of consistent data association.}
\end{figure}

Our contributions are the following. 
1) We propose a novel consensus-like algorithm for distributed data association.
2) We propose a decentralized version of the centralized state-of-art spectral method \cite{pachauri2013solving}.
3) We show that both methods provably converge,  do not depend on initialization, are parameter-free and guarantee global consistency.
%2) We shed light on the combinatorial nature of consistent data association by proving that even deciding whether a consistent data association exists is NP-complete.
4) We demonstrate the validity of proposed methods through both theoretical analysis and  experimentation with  synthetic and real  data.

The rest of this paper is organized as follows. Section \ref{sec:related} includes a brief review of prior works.
Section \ref{sec:preliminary} introduces the reader to notation and preliminaries  essential to understanding the main results of this paper. 
A rigorous problem formulation is the topic of Section \ref{sec:problemformul}.
The two proposed methods are presented in Sections \ref{sec:consensus} and \ref{sec:spectral} respectively.
%In Section \ref{sec:npcomplete}, we prove that a particular version of data association is NP-complete.
Finally, experimental evaluation is presented in Section \ref{sec:experiments}.

\section{RELATED WORK}\label{sec:related}

The cycle consistency has been explored in a variety of computer vision applications. For example, Zach et al. \cite{zach2010disambiguating} studied the problem of multi-view reconstruction and proposed to identify unreliable geometric relations (e.g. relative poses) between views by measuring the consistencies along a number of cycles. Nguyen et al. \cite{nguyen2011optimization} addressed the problem of finding point-wise maps among a collection of shapes and proposed to use the cycle consistency to identify the correctness of maps and replace the incorrect maps with the compositions of correct ones. Yan et al.  \cite{yan2015consistency,yan2016multi} considered the multiple graph matching problem and imposed the cycle consistency by penalizing the inconsistency during the optimization. Zhou et al. \cite{zhou2015flowweb} developed a discrete searching algorithm that maximizes the cycle consistency to improve the dense optical flow in an image collection. Instead of iteratively optimizing the pairwise matches, \cite{kim2012exploring,huang2012optimization,pachauri2013solving} showed that finding globally consistent associations from noisy pairwise associations can be formulated as a quadratic integer program and relaxed into a generalized Rayleigh problem and solved by the eigenvalue decomposition. More recently, Huang and Guibas \cite{huang2013consistent} and Chen et al. \cite{chen2014near} theoretically proved that the consistent associations could be exactly recovered from noisy measurements under certain conditions by convex relaxation and solved the problem by semidefinite programming. Zhou et al. \cite{zhou2015multi} proposed to solve the problem by rank minimization and developed a more efficient algorithm. All of these works addressed problem in a centralized setting, i.e., all pairwise measurements are available and optimized jointly.

In the robotics community, numerous approaches for associating elements of two sets  have been proposed and extensively used in single robot Simultaneous Localization And Mapping (SLAM) settings. These include but are not limited to Nearest Neighbor (NN) and Maximum Likelihood (ML) approaches \cite{kaess2009covariance,zhou2006multi},  Iterative Closest Point (ICP) \cite{besl1992method},  RANSAC  \cite{fischler1981random}, Joint Compatibility Branch and Bound (JCBB) \cite{neira2001data}, Maximum Common Subgraph
(MCS) \cite{bailey2001localisation}, Random Finite Sets (RFS) \cite{vo2005sequential}.
 More recently,  multi-robot localization and mapping algorithms have been developed. In almost all of them, the data association is addressed either in a pairwise fashion by directly generalizing data association techniques from the single robot case  \cite{cunningham2012fully,williams2002towards,fox2006distributed} or  by considering all associations jointly in a centralized manner \cite{indelman2014multi}.
Closer to our approach is the work by Aragues et al. \cite{aragues2011consistent} which is a decentralized method for finding inconsistencies  based on cycle detection. However, the proposed inconsistency resolution algorithm comes without any optimality or guarantees.
 
%Cunningham et al.  \cite{cunningham2012fully} proposed a fully decentralized method for multi-robot SLAM in which two maps are merged  using RANSAC and Delaunay triangulations. 
%In \cite{indelman2014multi} a centralized EM approach was proposed.  introduce a latent binary variable for each possible multi-robot data association,  local convergence due to the nature of EM but they consider all associations at once.
%In \cite{williams2002towards,fox2006distributed} distributed approach but they do map merging in a pairwise fashion.
%All the above methods either consider all associations jointly in a centralized setting or  consider one pairwise associations at once.

%%%%%%%%%%%%%%%%%%%%%%%%%%%%%%%%%%%%%%%%%%%%%%%%%%%%%%%%%%%%%%%%%%%%%%%%%%%%%%%%

%%%%%%%%%%%%%%%%%%%%%%%%%%%%%%%%%%%%%%%%%%%%%%%%%%%%%%%%%%%%%%%%%%%%%%%%%%%%%%%%

\section{PRELIMINARIES, NOTATION AND DEFINITIONS}
\label{sec:preliminary}

\subsection{Graph Theory}

In this subsection, we review some elementary facts from graph theory.
For an in depth analysis, we refer the reader to standard texts \cite{ godsil2013algebraic,mesbahi2010graph}.
An \notion{undirected graph} or simply a \notion{graph} is denoted by the pair $\graph= (\vset,\eset)$,
where $\vset =\{1,2,\ldots,n\}$ is the set of vertices and $\eset \subseteq [\vset]^2$ is the set of edges,
where $[\vset]^2$ denotes the set of unordered pairs of elements of $\vset$.
The neighborhood $\cNi$ of the vertex $i$ is  the subset of $\vset$ defined by $\cNi = \{ j \in \vset \ | \ \ijEu \}$.
A \notion{path} is a sequence ${i_0},{i_1},\ldots,{i_m}$ of distinct vertices such that $\{{i_{k-1}},{i_{k}}\} \in \eset$ for all $k =1,\ldots,m$.
A graph is \notion{connected} if there is a path between any two vertices.
Given a graph $\graph= (\vset,\eset)$, its \notion{adjacency matrix} is  defined by
\begin{equation}
[A(\graph)]_{ij} =
\left\{\begin{array}{ll}
 1,  &\text{if } \ijEu \\
 0,  &\text{otherwise }
 \end{array}
 \right.
\end{equation}
The \notion{degree matrix} $\Delta(\graph)$  is the diagonal matrix   such that $[\Delta(\graph)]_{ii} = \card{\cNi}$,
where $\card{\cdot}$ denotes the cardinality of a set.

A \notion{directed graph} or \notion{digraph} is denoted by the pair $\graph= (\vset,\eset)$,  where $\eset \subseteq \vset \times \vset$.
A \notion{weighted digraph} $\graph = (\vset,\eset,w)$ is a graph along with a function $\mapdef{w}{\eset}{\reals_+}$. The adjacency matrix of a weighted digraph is defined by
\begin{equation}
[A(\graph)]_{ij} =
\left\{\begin{array}{ll}
 w(j,i) ,  &\text{if } \jiE \\
 0,  &\text{otherwise }
 \end{array}
 \right.
\end{equation}
Intuitively, if $[A(\graph)]_{ij}>0$ there is information flow from vertex $j$ to vertex $i$.  The neighborhood $\cNi$ of the vertex $i$ is  the subset of $\vset$ defined by $\cNi = \{ j \in \vset \ | \ \jiE \}$.
The \notion{degree matrix} $\Delta(\graph)$
\begin{equation}
[\Delta(\graph)]_{ii} = d_{\textrm{in}}(i) \doteq \sum_{ j \in \cNi } [A(\graph)]_{ij}
\end{equation}
A \notion{directed path} is a sequence ${i_0},{i_1},\ldots,{i_m}$ of distinct vertices such that $({i_{k-1}},{i_{k}}) \in \eset$ for all $k =1,\ldots,m$.
A digraph is \notion{strongly connected} if there is a directed path between any two vertices.
A digraph is a \notion{rooted out-branching tree} if it has a vertex to which all other vertices are path connected and does not contain any cycles.
A digraph is \notion{balanced} if the in-degree $d_{in}(i)$ and the out-degree   $d_{\textrm{out}}(i) \doteq \sum_{ (i,j) \in \eset } [A(\graph)]_{ji}  $ are equal for all $i \in \vset$.
 From this point and on, all graphs under consideration are assumed to be directed unless stated otherwise. 
 The \notion{graph Laplacian} $L(\graph)$ is defined as
\begin{equation}
L(\graph) = \Delta(\graph) - A(\graph)
\end{equation}
 By construction, $L(\graph) \mathbf{1} =0$. A digraph on $n$ vertices contains a rooted-out branching if and only the rank of its Laplacian is $n-1$. % see prop 3.8 page 51 of Meshbahi
%Moreover,  it is positive semidefinite if $A$ is symmetric.
%A graph $\graph$ is connected if and only if $\lambda_2(L(\graph) )>0$.
%Using Gersgorin discs theorem \cite{horn2012matrix}, we can see that all the eigenvalues of the adjacency matrix lie in the inteval $[-\degmax,\degmax]$.
A matrix closely related to the graph Laplacian is defined by
\begin{equation}
\label{eq:definiingF}
F(\graph) \doteq (I+\Delta(\graph))^{-1}(I+ A(\graph) )
\end{equation}

\subsection{Stochastic Matrices and Permutations}

Stochastic matrices and their properties have been well studied in the area of distributed dynamical systems \cite{jadbabaie2003coordination,olfati2007consensus,bertsekas1989parallel} and
in probability theory in the  context of computing  invariant distributions of Markov chains \cite{asmussen2008applied,bertsekas1989parallel}.
A nonnegative square matrix is   \notion{stochastic} if all its row sums are equal to 1.
The spectral radius $\rho(A)$ of a stochastic matrix $A$ is equal to 1 and it is an eigenvalue of $A$ since $A \mathbf{1}=\mathbf{1}$.
The existence and the form of the limit $\lim_{k \rightarrow \infty}A^k$ depends on the multiplicity of  $1$ as an eigenvalue.
A nonegative square matrix $A$ is \notion{irreducible} if its associated graph \cite{horn2012matrix} is strongly connected.
An irreducible stochastic matrix is  \notion{primitive} if it  has only one eigenvalue of maximum modulus.

A nonnegative square matrix  is \notion{doubly stochastic} matrix if both its row sums and column sums are equal to 1.
A doubly stochastic matrix is a \notion{permutation matrix} if its elements are either 0 or 1.
The sets of stochastic matrices and doubly stochastic matrices are compact convex sets.
This is a rather useful property since it implies closure under convex combinations.
For more details, regarding the theory of stochastic matrices, we refer the reader to \cite{horn2012matrix}.

%A matrix is irreducible if it is not similar via a permutation to a block upper triangular matrix.
%If a square matrix $A$ is nonnegative, then $A$ is \textit{primitive} if and only if $A^m >0$ for some $m \geq 1$.
%A stochastic matrix $A$ such that $\lim_{k\rightarrow \infty} A^k$ has rank 1 is called \textit{ergodic}.
%Primitive stochastic matrices are ergodic.

%Some facts from \cite{horn2012matrix}.
%Given a nonnegative $n \times n$ matrix $A$, the \notion{associated digraph $\Gamma(A)$} is defined as follows: it has $n$ vertices and there is a (directed) edge from vertex $i$ to vertex $j$ if and only if $[A]_{ij} >0$.  There is a directed pathfrom $i$ to $j$ of length $m$ if and only if $[A^m]_{ij} >0$.
%Moreover, $A$ has the property SC is the associated directed graph $\Gamma(A)$ is strongly connected. Given  a nonnegative $A \in \real{n\ \times n }$ the following are equivalent
%\begin{itemize}
% \item $A$ is irreducible
% \item $(I+A)^{n-1} > 0$
% \item $\Gamma(A)$ is strongly connected
%\item $A$ has the property SC
%\end{itemize}

Let $[n] \doteq \{1,2,\ldots,n \}$ for some positive integer $n$.
A mapping  $\mapdef{\pi}{[n]}{[n]}$ is a \notion{permutation} of $[n]$ if it is bijective.
The set of all permutations of $[n]$ forms a group under composition, termed the \notion{symmetric group} $\symgroup{n}$.
A permutation $\pi \in \symgroup{n}$ is represented
\footnote{A representation of a group $G$ on $\reals^n$ is a map $\mapdef{\rho}{G}{\GL(n)}$ satisfying $\rho(g_1 g_2) = \rho(g_1) \rho(g_2)$.}
by a $n \times n$ permutation matrix $\Pi$ defined by
\begin{equation}
[\Pi]_{ij} =
 \left\{\begin{array}{ll}
 1,  &\text{if } \pi(j)=i \\
 0,  &\text{otherwise }
 \end{array}
 \right.
\end{equation}
or equivalently $\Pi e_j = e_{\pi(j)}$,  where $e_j$ is the $j$th canonical basis vector.
%The set of $n \times n$ permutation matrices is a group under matrix multiplication.
%The inverse of a permutation matrix is given by its transpose, \ie $\Pi^{-1} = \Pi^T$.
The simplest choice for a distance  on $\symgroup{n}$ is given by
\begin{equation}
\label{eq:distanceperm}
d(\pi_1,\pi_2)
= d(e,\pi_1^{-1}\pi_2)
 \doteq  n - \metric{\Pi_1}{\Pi_2}
\end{equation}
where  $\metric{A}{B}\doteq \trace(A^TB)$, $e$ is the identity map and $\Pi_1,\Pi_2$ are the matrix representations of the permutations $\pi_1$ and $\pi_2$ respectively. 
The  distance function defined above is simply the number of labels assigned differently by permutations  $\pi_1$ and $\pi_2$.
With a slight abuse of notation, we write $\Pi \in \symgroup{n}$ meaning that $\Pi$ is an $n\times n$ permutation matrix.

\subsection{Consensus Algorithms}

%%%%%%%%%%%%%%%%%%%%%%%%%%%%%%%%%%%%%%%%%%%%%%%%%%%%%%%%%%%%%%%%%%%%%%%%%%

Consensus algorithms have been extensively studied in the control community \cite{bertsekas1989parallel,jadbabaie2003coordination,olfati2007consensus}. In its simplest form, a consensus algorithm is a decentralized protocol in which the agents, modeled as vertices of a graph,  try to reach agreement by communicating only with their neighbors. More formally, let $x_i(t) \in \reals$ denote the state of agent $i$ at time time $t$. The, the simplest consensus protocol is given by
\begin{equation}
x_i(t+1) = \sum_{j \in \cNi \cup {i}} a_{ij} x_j(t)
\end{equation}
where $a_{ij} \geq 0$ and $ \sum_{j} a_{ij}=1$.
A popular consensus protocol \cite{olfati2007consensus}, which converges to the average of the initial values provided that $\graph $ is balanced, is described by
\begin{equation}
\label{eq:standarconsensusprotocol}
\mathbf{x}(t+1) = (I - \epsilon L(\graph)) \mathbf{x}(t)
\end{equation}
where $\mathbf{x}=[x_1 , \ldots , x_n]^T$ and $0 < \epsilon < 1/\max_i [\Delta(\graph)]_{ii}$.

\medskip

%%%%%%%%%%%%%%%%%%%%%%%%%%%%%%%%%%%%%%%%%%%%%%%%%%%%%%%%%%%%%%%%%%%%%%%%%%
\section{PROBLEM FORMALIZATION}
\label{sec:problemformul}

In this section, we formalize the problem of consistent data association.
We assume there are $n$ sensors observing $m$ targets. 
For instance, assume we have a set of $n$ cameras observing a  scene in the world described by a set of $m$ feature points. 
Sensors communicate only with a subset of all sensors.
Communication constraints  between sensors are encoded by the \textit{sensor graph}. 
The sensor graph is the digraph $G=(\vset,\eset)$ where $\vset=\{1,2,\ldots,n\}$ and $(i,j) \in \eset$ if there is information flow from sensor $i$ to $j$.
The pairwise association $\pi_{ij} \in \symgroup{m}$ is defined as follows:
we have that $\pi_{ij}(l)=k$ if the $l$th target in $j$th sensor corresponds to the $k$th target in $i$th sensor.
Observe that the pairwise associations $\pi_{ij}$ can be written as  $ \pi_{ij} = \pi_i \circ \pi_j^{-1}$, 
where $\pi_i \in \symgroup{m}$ is the mapping from the labels of sensor $i$ to some global labels, termed the ``universe of features'' in some works \cite{chen2014near,zhou2015multi}.  
This is analogous to the absolute position in Euclidean spaces in localization problems.
We denote by $\tpi_{ij} \in \symgroup{m}$ the, possibly erroneous, estimated pairwise association between sensor $i$ and $j$
and by $\tPi_{ij}$ the corresponding matrix representation. Moreover, let $\Pi_i \doteq \rho(\pi_i)$.

%Let $\Pi_{ij} \doteq \rho(\pi_{ij})$ be the corresponding matrix representation of $\pi_{ij}$.
%\begin{equation} \Piij e_l = e_{\pi_{ij}(l)}  \end{equation}
%if label $k$ in $i$th sensor is associated with label $l$ in the $j$th sensor and $0$ otherwise. 
% \[  \Piij e_l = \Pii \Piij^{-1} e_l = \Pii e_{\pi^{-1}_j(l)} = e_{\pi_i (\pi^{-1}_j(l))}  = e_{\pi_{ij}(l)}   \]

Related to the sensor graph is the \textit{data association graph} $\cD = (\vset_{\cD},{\eset}_{\cD},w_\cD)$, where $\vset_{\cD} = \vset \times \{ 1,2,\ldots,m\}$. 
There is an edge from $(i,k)$ to $(j,l)$ if and only if $(i,j) \in \eset$ and $[\tPiij]_{kl} >0$. 
The corresponding edge weight is simply equal to $[\tPiij]_{kl}$.

%\remark{In some cases,  $\tPiij$'s may be doubly stochastic matrices instead of permutation matrices. The}

First, we need a precise definition of \notion{consistency}.
\begin{definition}[Consistency]
A set of pairwise associations $\{ \tpi_{ij} \}_{(i,j)\in \eset}$ is \textit{consistent} if
\begin{equation}
\label{eq:consistency}
\tpi_{ij} \circ \tpi_{jk} = \tpi_{ik}
\end{equation}
for all valid indices $i,j,k$. A set of labels $\{\pi_i\}_{i=1}^n$ is \notion{consistent} with respect to some consistent pairwise associations  $\{ \tpi_{ij} \}_{(i,j)\in \eset}$, if
\begin{equation}
 \tpi_{ij} = \pi_i \circ \pi_j^{-1} , \qquad \forall (i,j) \in \eset
\end{equation}

%Equivalently in matrix notation we have
%\begin{equation}
%\label{eq:consistencymatrix}
%\Pi_{ij}  \Pi_{jk} = \Pi_{ik}
%\end{equation}
%for all $i,j,k \in \{1,2,\ldots,n\}$.
\end{definition}

\medskip

Based on the definition of consistency, the  problem of consistent data association is naturally defined as follows. 
\begin{definition}[\bf{Consistent data association}]
 Given  pairwise associations $\{\tpi_{ij}\}_{(i,j)\in \eset}$,  find labels $\pi_1,\ldots,\pi_n \in \symgroup{m}$, such that
 \begin{equation}
 \label{eq:problemdef}
\tpi_{ij} = \pi_{i} \circ \pi_j^{-1} , \qquad \forall \ (i,j) \in \eset
 \end{equation}
\end{definition}

\medskip
\remark{Under the presence of noise, it might not be possible to find labels $\{\pi_i \}_{i=1}^n$ satisfying \eqref{eq:problemdef} exactly.
Therefore, in practice we are looking for labels $\{\pi_i \}_{i=1}^n$ satisfying \eqref{eq:problemdef} as much as possible according to some criterion.
}

\begin{comment}
\medskip
Next, we have the consistency condition \eqref{eq:consistency}  in terms of the representations of the pairwise associations.
\begin{lemma}[Rank constraint for consistency \cite{huang2013consistent}] 
Given pairwise associations $\{ \tpi_{ij} \}_{(i,j)\in \eset}$, define the block matrix $\mathbf{P}$ by  $[\mathbf{P}]_{ij} = \tPiij\doteq \rho(\tpi_{ij})$.
The set of pairwise associations $\{ \tpi_{ij} \}_{(i,j)\in \eset}$ is \textit{consistent} if and only if
\begin{equation}
\label{eq:rankconditions}
\mathbf{P} \succeq 0, \quad [\mathbf{P}]_{ii} = I ,\quad  \rank (\mathbf{P}) = m
\end{equation}
\end{lemma}
\medskip \noindent
To see why \eqref{eq:rankconditions} holds, assume we are given $\mathbf{P}$ satisfying \eqref{eq:rankconditions}.
Then $\mathbf{P}$ can be written as
\begin{equation}
\mathbf{P} =
\begin{bmatrix}
\Pi_1 \\ \vdots \\ \Pi_n
\end{bmatrix}
\cdot\begin{bmatrix}
\P_1^T & \cdots & \Pi_n^T
\end{bmatrix}
\end{equation}
where each $\Pi_i$ is an $m \times m$ orthogonal matrix.
We get that consistency is satisfied since $
\tPi_{ij} \tPi_{jk} =  P_{i}P_{j}^T P_{j} P_{k}^T = P_{i} P_{k}^T = \tPi_{ik}
$.  
\end{comment}

\begin{figure}
\label{fig:consistentdata}
\centering
\scalebox{0.9}{
\begin{tikzpicture}[thick,
  every node/.style={draw,circle},
  fsnode/.style={},
  gsnode/.style={fill=myblue},
  ssnode/.style={fill=mygreen},
  dsnode/.style={fill=myred},
  every fit/.style={ellipse,draw,inner sep=-2pt,text width=1.2cm}
]

% the vertices of U
\begin{scope}[start chain=going below,node distance=7mm]
  % \node[fsnode,on chain] (f\i) [label=left: \i] {};
  \node[fsnode,on chain,fill=myblue] (f1) [label={[xshift=-0.0cm, yshift=-0.9cm]$1$}] {};
  \node[fsnode,on chain,fill=myred] (f2)  [label={[xshift=-0.0cm, yshift=-0.9cm]$2$}]{};
  \node[fsnode,on chain,fill=mygreen] (f3)  [label={[xshift=-0.0cm, yshift=-0.9cm]$3$}] {};
\end{scope}

% the vertices of V
\begin{scope}[xshift=2cm,yshift=0.0cm,start chain=going below,node distance=7mm]
  \node[ssnode,on chain,fill=myblue] (s1) [label={[xshift=-0.0cm, yshift=-0.9cm]$1$}] {};
  \node[ssnode,on chain,fill=myred] (s2)  [label={[xshift=-0.0cm, yshift=-0.9cm]$2$}]{};
  \node[ssnode,on chain,fill=mygreen]  (s3)  [label={[xshift=-0.0cm, yshift=-0.9cm]$3$}] {};
\end{scope}

% the vertices of W
\begin{scope}[xshift=4cm,yshift=0.0cm,start chain=going below,node distance=7mm]
  \node[dsnode,on chain,fill=myblue] (d1) [label={[xshift=-0.0cm, yshift=-0.9cm]$1$}] {};
  \node[dsnode,on chain,fill=myred] (d2)  [label={[xshift=-0.0cm, yshift=-0.9cm]$2$}]{};
  \node[dsnode,on chain,fill=mygreen]  (d3)  [label={[xshift=-0.0cm, yshift=-0.9cm]$3$}] {};
\end{scope}

% the set U
\node [black,fit=(f1) (f3),label={[xshift=0.0cm, yshift=0.0cm]$S_1$}]  {};
% the set V
\node [black,fit=(s1) (s3),label={[xshift=0.0cm, yshift=0.0cm]$S_2$}]  {};
% the set W
\node [black,fit=(d1) (d3),label={[xshift=0.0cm, yshift=0.0cm]$S_3$}]  {};

% the edges
\draw (f1) -- (s1);
\draw (f2) -- (s2);
\draw (f3) -- (s3);
\draw (s1) -- (d1);
\draw (s2) -- (d2);
\draw (s3) -- (d3);
\draw (f1) to [out=20,in=160] (d1);
\draw (f2) to [out=20,in=160] (d2);
\draw (f3) to [out=20,in=160] (d3);
\end{tikzpicture}
}
\vspace{0.1cm}
\\
\scalebox{0.9}{
\begin{tikzpicture}[thick,
  every node/.style={draw,circle},
  asnode/.style={fill=myblue},
  bsnode/.style={fill=mygreen},
  csnode/.style={fill=myred},
  every fit/.style={ellipse,draw,inner sep=-2pt,text width=1.2cm}
]

% the vertices of U
\begin{scope}[start chain=going below,node distance=7mm]
  \node[asnode,on chain,fill=myblue] (a1) [label={[xshift=-0.0cm, yshift=-0.9cm]$1$}] {};
  \node[asnode,on chain,fill=myred] (a2)  [label={[xshift=-0.0cm, yshift=-0.9cm]$2$}]{};
  \node[asnode,on chain,fill=mygreen] (a3)  [label={[xshift=-0.0cm, yshift=-0.9cm]$3$}] {};
\end{scope}

% the vertices of V
\begin{scope}[xshift=2cm,yshift=0.0cm,start chain=going below,node distance=7mm]
  \node[bsnode,on chain,fill=myblue] (b1) [label={[xshift=-0.0cm, yshift=-0.9cm]$1$}] {};
  \node[bsnode,on chain,fill=myred] (b2)  [label={[xshift=-0.0cm, yshift=-0.9cm]$2$}]{};
  \node[bsnode,on chain,fill=mygreen] (b3)  [label={[xshift=-0.0cm, yshift=-0.9cm]$3$}] {};
\end{scope}

% the vertices of W
\begin{scope}[xshift=4cm,yshift=0.0cm,start chain=going below,node distance=7mm]
  \node[csnode,on chain,fill=myblue] (c1) [label={[xshift=-0.0cm, yshift=-0.9cm]$1$}] {};
  \node[csnode,on chain,fill=myred] (c2)  [label={[xshift=-0.0cm, yshift=-0.9cm]$2$}]{};
  \node[csnode,on chain,fill=mygreen] (c3)  [label={[xshift=-0.0cm, yshift=-0.9cm]$3$}] {};
\end{scope}

% the set U
%\node [myblue,fit=(f1) (f3),label=above:$U$] {};
\node [black,fit=(a1) (a3)] {};
% the set V
\node [black,fit=(b1) (b3)] {};
% the set W
\node [black,fit=(c1) (c3)] {};
% the edges
\draw (a1) -- (b1);
\draw (a2) -- (b3);
\draw (a3) -- (b2);
\draw (b1) -- (c1);
\draw (b2) -- (c2);
\draw (b3) -- (c3);
\draw (a1) to [out=30,in=150] (c1);
\draw (a2) to [out=20,in=160] (c2);
\draw (a3) to [out=-30,in=210] (c3);
\end{tikzpicture}
}
\caption{Example with $n=3$ sensors $S_1,S_2,S_3$ observing  $m=3$ targets. Top: consistent data association. Bottom: inconsistent data association since $\pi_{12}\circ \pi_{23}(2) =3 $ but $\pi_{31}(2)=2$.}
\end{figure}
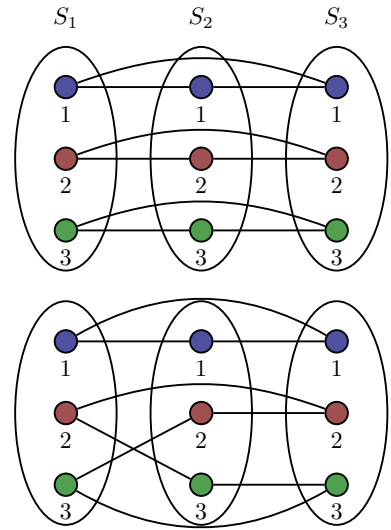

%%%%%%%%%%%%%%%%%%%%%%%%%%%%%%%%%%%%%%%%%%%%%%%%%%%%%%%%%%%%%%%%%%%%%%%%%%%%%%%%

\section{DISTRIBUTED AVERAGING}

\label{sec:consensus}

In the classic consensus algorithms, each agent updates his estimate of a collective quantity by taking convex combinations of the estimates of his neighbors. The problem at hand is  different: we do not want all $\pi_i$'s to converge to the same value, rather to converge to a value satisfying consistency.
The first obstacle one encounters is the combinatorial nature of the problem at hand.
Therefore, we relax the domain of problem and work with doubly stochastic matrices instead of permutation matrices.
This is a natural choice since the convex hull of the set of $m \times m$ permutation matrices is exactly the set of $m \times m$ doubly stochastic matrices.

 \subsection{Algorithm}
 
Based on the previous discussion, we propose the following update rule:
\begin{equation}
\label{eq:update}
\Pii(t+1) =    \dfrac{1}{\card{\cNi}+1} \bigl(  \Pii(t) + \sum_{j \in \cNi} \tPiij \Pij(t)\bigr)
\end{equation}

Let $\mathbf{\Pi} = [\Pi_1^T,  \ldots , \Pi_n^T]^T$. The update rule \eqref{eq:update} can be collectively written as
\begin{equation}
\mathbf{\Pi}(t+1) = F(\cD) \mathbf{\Pi}(t)
\end{equation}
where $\cD$ is the data association graph and $F(\cD)$ is defined in \eqref{eq:definiingF}. The convergence properties of the update rule \eqref{eq:update}  depend solely on $\lim_{k \rightarrow \infty} F(\cD)^k$ and $\mathbf{\Pi} (0)$. These properties will be presented in Section \ref{sec:properties}. 
%\todo{Write something intuitive here to prepare for the results of the next subsection}

After the convergence of Protocol \eqref{eq:update}, we discretize the solution by solving  the assignment problem, formally defined by
\begin{equation}
\begin{aligned}
& \underset{\Pi_i}{\text{maximize}}
& & \metric{\tPi_i}{\Pi_i} \\
& \text{subject to}
& &  \Pi_i \in \symgroup{m}
\end{aligned}
\end{equation}
where $\tPi_i$ is the result of Protocol \eqref{eq:update}. The assignment problem  can be solved using the Hungarian  algorithm \cite{kuhn1955hungarian} in $O(m^3)$ time.

\remark{The update rule \eqref{eq:update} can be seen as a discrete-time system with state $\Pi_i$ and control input $U_i$, \ie
\begin{align}
\Pii(t+1) &=   \Pii(t) +  U_i(t) \\
  U_i(t) &= \dfrac{1}{\card{\cNi}+1} \sum_{j \in \cNi}  \bigl(  \tPiij \Pij(t)-\Pii(t)  \bigr)
\end{align}
}

\remark{
(\textit{Intrinsic ambiguity})
Let $\{\Pi_i\}_{i=1}^n \in (\symgroup{m})^n$ be a set of consistent labels. Then, for any  $\Pi_0 \in \symgroup{m}$, $\{\Pi_i\Pi_0\}_{i=1}^n$ are consistent as well. This intrinsic ambiguity is analogous to the ambiguity of the global reference frame in rotation localization \cite{tron2014distributed}. To remove this ambiguity, one can fix the label of one sensor, say the first sensor, by modifying the sensor graph $\graph$ so that the first sensor has only outgoing edges resulting in $\Pi_1(t)  = \Pi_1(0)  \in \symgroup{m}$ for all $t \geq 0$. 
This approach is necessary in the presence of noise and we will refer to it as Protocol \eqref{eq:update} with a distinguished vertex. 
}

\remark{
(\textit{Equivalence to consensus})
In the noiseless case, we can write $\tPi_{ij} = \Pi_{i0} \Pi_{j0}^{-1}$ for some consistent set of labels $\{\Pi_{i0}\}_{i=1}^n \in (\symgroup{m})^n$ . 
By defining new variables $\Pi_i'\doteq \Pi_{i0}^{-1} \Pi_i$, we obtain
\begin{equation}
\label{eq:consensusequiv}
\Pii'(t+1) =    \dfrac{1}{\card{\cNi}+1} \bigl(  \Pii'(t) + \sum_{j \in \cNi}  \Pij'(t)\bigr)
\end{equation}
which is a component-wise Vicsek model \cite{vicsek1995novel,jadbabaie2003coordination}.
}

\subsection{Properties}\label{sec:properties}

First of all, as the sets of stochastic and doubly stochastic matrices are convex and thus, closed under convex combinations, we naturally have the following lemma.
\begin{lemma}
If all $\tPi_{ij}$ and $\Pi_i(0)$ are stochastic then, $\Pi_i(t)$ is stochastic for all $t \geq 0$. Similarly, if all $\tPi_{ij}$ and $\Pi_i(0)$ are doubly stochastic then, $\Pi_i(t)$ is doubly stochastic for all $t \geq 0$.
\end{lemma}

Next, we show that in the noiseless case Protocol \eqref{eq:update} with a distinguished vertex converges to a consistent solution under mild coniditions on the sensor graph.  To prove this, we need the following lemma:
\begin{lemma}
\label{lem:perronspecial}
Given a digraph $\graph$ that contains a rooted-out branching tree, we have that the matrix $F(\graph)$ as defined in \eqref{eq:definiingF} satisfies: (a) $\rho(F(\graph)) = 1$, (b)  $1$ is an algebraically simple eigenvalue of $F(\graph)$ with corresponding eigenvector $\mathbf{1}/\sqrt{n}$ and (c) the limit $F(\graph)$ exists and is given by
\begin{equation}
\lim_{k \rightarrow \infty} F(\graph)^k = \mathbf{1} c^T, 
\end{equation}
where $c \geq 0$ and $c^T \mathbf{1} = 1$.
\end{lemma}
A proof of lemma \ref{lem:perronspecial} can be found in Appendix \ref{app:proofoflemmaperronspecial}. Next, we present the theorem of convergence of Protocol \ref{eq:update} to a consistent labeling from arbitrary initialization. A proof of theorem \ref{thm:covergence_noiseless} is presented in Appendix \ref{app:app_covergence_noiseless} and heavily relies on lemma \ref{lem:perronspecial}.
\begin{thm}
\label{thm:covergence_noiseless}
If the sensor graph $\graph$ contains a rooted-out branching and the pairwise associations are noiseless, then the consensus protocol with the distinguished vertex converges to a set of consistent labels up to a global permutation, that is
\begin{equation}
\lim_{t \rightarrow \infty} \Pi_i(t) = \Pi_{i0} \Pi_0
\end{equation}
where $\{\Pi_{i0}\}_{i=1}^n$ is a consistent set of labels (with respect to $\tPi_{ij}$) and $\Pi_0$ is an arbitrary permutation.
Furthermore, if vertex $1$ is the distinguished vertex, then $\Pi_0 = \Pi_{10}^T \Pi_1(0)$.
\end{thm}

However, in the general case, pairwise associations will contain noise. 
Therefore, we should prove convergence without assuming perfect associations.
We generalize lemma \ref{lem:perronspecial} for the   case of interest. Consider a digraph $\cD$ with $n$ vertices and $m < n$ distinguished vertices that have only outgoing edges. Assume that for every non-distinguished  vertex there is a (directed) path from at least one distinguished vertex.  Then, the succeeding lemma provides us with the limit of $F(\cD)$.
\begin{lemma}
\label{lem:perronspecial2}
The matrix $F(\cD)$ as defined in \eqref{eq:definiingF} satisfies: (a) $\rho(F(\cD)) = 1$, (b) $F(\cD)$ has  $1$ as an eigenvalue with algebraic and geometric multiplicities equal to the number of distinguished vertices, (c) it asymptotically converges to a limit of the form
\begin{equation}
\lim_{k \rightarrow \infty} F(\cD)^k = \begin{bmatrix}
I_m & 0 \\
\widetilde{F}_{21} & 0
\end{bmatrix}
\end{equation}
\end{lemma}
A proof of lemma \ref{lem:perronspecial2} is presented in Appendix \ref{app:proofperronspeical2}. An immediate consequence of lemma \ref{lem:perronspecial2} is the following theorem.
\begin{thm}
\label{thm:covergence_noise}
If the sensor graph $\graph$ contains a rooted-out branching, then the consensus protocol with a distinguished vertex asymptotically converges to a solution that depends only on   $\{ \tPi_{ij}\}_{(i,j)\in \eset}$ and on the initial value of the labels of the distinguished vertex.
\end{thm}

\remark{
Although Protocol \eqref{eq:update} is not guaranteed to converge to the true solution in the presence of noise, it is verified experimentally (see Section \ref{sec:experiments}) that for moderate levels of noise the true solution is still recovered.
}

%\todo{Write how efficient it is especially when $\tPi_{ij}$ is a permutation}

\section{DISTRIBUTED ORTHOGONAL ITERATION}

\label{sec:spectral}

In this section, we present a decentralized implementation of the spectral method proposed by Pachauri et al. \cite{pachauri2013solving}. 
First, we briefly review the main concept of \cite{pachauri2013solving}. 
Based on the discussion in Section \ref{sec:problemformul}, the multi-sensor data association can be cast as the following optimization problem:
\begin{equation}
 \underset{\pi_1,\pi_2,\ldots,\pi_n \in \symgroup{m} }{ \textrm{minimize}}  \sum_{ (i,j) \in \eset } d ( \widetilde{\pi}_{ij},\pi_i\circ \pi_j^{-1} )
\end{equation}
which is equivalent to
\begin{equation}
\label{eq:optrelax1}
 \underset{\Pi_1,\Pi_2,\ldots,\Pi_n \in \symgroup{m} }{ \textrm{maximize}}  \sum_{ (i,j) \in \eset }  \trace( \Pi_i^T \tPi_{ij} \Pi_j)
\end{equation}
In the above formulation, each $\Pi_i$ is a permutation matrix which makes the problem computationally intractable. However, if the individual constraints  
are relaxed  into a single collective constraint of the form $\mathbf{\Pi}^T\mathbf{\Pi} = m I_m$ as proposed in \cite{pachauri2013solving}, the relaxed problem becomes computationally tractable since it becomes a Rayleigh quotient problem:
\begin{equation}
\label{eq:optrelax2}
\begin{aligned}
& \underset{\mathbf{\Pi}}{\text{maximize}}
& & \trace(\mathbf{\Pi}^T \mathbf{P} \mathbf{\Pi} ) \\
& \text{subject to}
& &  \mathbf{\Pi}^T\mathbf{\Pi} = m I_m
\end{aligned}
\end{equation}
where $[\mathbf{P}]_{ij} = \tPi_{ij}$. It is a well-known fact that the solution of problem (\ref{eq:optrelax2}) is given by the matrix having as columns the $m$ leading eigenvectors of $\mathbf{P}$ scaled by a factor of $\sqrt{m}$. Numerically, the optimal solution in a centralized setting can be computed using  Orthogonal Iteration \cite{golub2012matrix} (see Algorithm \ref{alg:OI}).
Intuitively,  each iteration consists of two main steps: a power step and an orthogonalization step based on $QR$-decomposition.
 For details on the convergence of Orthogonal Iteration we refer the reader to \cite{golub2012matrix} and references therein.

\begin{algorithm}
\caption{Orthogonal Iteration}
\label{alg:OI}
\begin{algorithmic}[1]
\Require Input matrix $\mathbf{P}$, initial $\mathbf{\Pi}(0)$
%\State Initialize $x(0)$.
\For{ $t=0,\ldots,N$}
\State $\mathbf{Y}(t)= \mathbf{P} \mathbf{\Pi}(t)$
\State $\mathbf{Q}(t) \mathbf{R}(t) = \mathbf{Y}(t)$
\State $\mathbf{\Pi}(t+1) =\mathbf{Q}(t)$
\EndFor
\end{algorithmic}
\end{algorithm}

The Orthogonal Iteration method is not readily decentralizable since the $QR$-decomposition step requires information from the entire sensor graph.
Nevertheless,  as observed by Kempe and McSherry \cite{kempe2004decentralized}, if $Y = [Y_1^T, \ldots,Y_n^T]^T$ where each sub-matrix $Y_i$ is available at sensor $i$,  we obtain $\mathbf{R}^T \mathbf{R}= \mathbf{Y}^T\mathbf{Y} = \sum_{i=1}^n Y_i^TY_i$. 
If  the sum $\sum_{i=1}^n Y_i^TY_i$ can be computed in a decentralized manner, then we can recover $\mathbf{R}(t)$ using Cholesky factorization.
Let $Z_i \doteq n Y_i^TY_i$ and $\mean{Z} \doteq \frac{1}{n} \sum_{i=1}^n Z_i = \mathbf{R}^T \mathbf{R}$. 
 Based on this observation, if all sensors know $n$ and solve a consensus problem for computing the average $\mean{Z} = \mathbf{R}^T \mathbf{R}$,  
then the orthogonalization step can be performed locally using  Cholesky decomposition. 
This is summarized in Algorithm \ref{alg:DOI} where $\mean{Z}_i$ denotes the estimate of $\mean{Z}$ by sensor $i$. 
Alternatively, the  sum $\sum_{i=1}^n Y_i^TY_i$ can be computed using the Push-Sum algorithm \cite{kempe2003gossip} as in \cite{kempe2004decentralized}.

%However, instead of the Push-Sum algorithm \cite{kempe2004decentralized} for computing, we use a standard  consensus protocol given in \eqref{eq:standarconsensusprotocol}.

\remark{ 
Kempe and McSherry \cite{kempe2004decentralized} have shown that if the error in the computation of $R_i$ can be made arbitrarily small in finite iterations, then Distributed Orthogonal Iteration converges asymptotically an invariant subspace of the input matrix.
A rigorous analysis on the number of iterations of the inner  consensus is the subject of future work.
}

\begin{algorithm}
\caption{Distributed Orthogonal Iteration}
\label{alg:DOI}
\begin{algorithmic}[1]
\Require Input associations $\{\tPi_{ij} \}_{(i,j)\in \eset}$, initial $\{\Pi_i(0)\}_{i \in \vset}$
%\State Initialize $x(0)$.
\For{ $t=0,\ldots,N$}
\State $Y_i(t)= \sum_{j \in \cNi \cup \{ i\}}\tPi_{ij} \Pi_j(t)$
\State Estimate $\mean{Z}_i(t)$ by consensus.
%\State $\mean{Z}_i(t) = R_i(t)^T R_i(t)$ 
\State $R_i(t) = \text{chol}(\mean{Z}_i(t) )$
\State $\Pi_i(t+1) = Y_i(t) R_i(t)^{-1}$
\EndFor
\end{algorithmic}
\end{algorithm}

\remark{Due to the nonlinear dependence of $Z_i$ on $Y_i$, dynamic consensus algorithms \cite{spanos2005dynamic} cannot be readily incorporated into the current approach.}

The solution obtained by the Orthogonal Iteration is optimal up to an arbitrary $m \times m$ orthogonal transformation. That is, if $\mathbf{\Pi}^\star$ is an optimal solution of   \eqref{eq:optrelax2}, then any matrix of the form $\mathbf{\Pi}^\star Q$  with $Q \in \OO{m}$ is  optimal as well.
Finally, the corrective orthogonal transformation is computed by solving the following  Orthogonal Procrustes \cite{golub2012matrix} problem:
\begin{equation} 
 \underset{Q \in \OO{m} }{ \textrm{minimize}}  \  \norm{I-\Pi_1 Q}_F^2
\end{equation}
It is well known \cite{golub2012matrix} that if $\Pi_1 = USV^T$ is an SVD of $\Pi_1$, then $Q^\star = V U^T$. As a final step, we apply the corrective transformation $Q$ to all remaining $\Pi_i$, $i=2,\ldots,n$ and then find a discrete solution by the Hungarian algorithm.

\section{EXPERIMENTS}

\label{sec:experiments}

\subsection{Synthetic data}

\label{subsec:experiments_synth}

First, we experiment with synthetic data to validate the accuracy of proposed methods across different problem settings. We generate instances with $m=50$ targets and $n\in \{20,50,100\}$ sensors. As a performance criterion we use the accuracy of the obtained labels,
which is simply defined as the percentage of correct labels. 
We vary the percentage of erroneous associations in each pairwise association from  10\% to 90\%  to test the robustness of each method against the presence of outliers.
We first experiment  with a fully connected graph, see Fig. 3 (a)-(c),  and then we remove half of the edges, see Fig. 3 (d)-(f).
In the case of a fully connected graph, the consensus protocol achieves exact recovery of the labels for up to approximately 40\% of outliers. 
The spectral method clearly outperforms the consensus protocol and is able to achieve exact recovery for up to 80\%-90\% of outliers.
By reducing either the number of edges in the sensor graph or  the number of sensors, the accuracy of both methods slightly drops.
As expected,  the benefit of using joint matching in accuracy increases with the size of the network and the number of  edges.

\def \figscale{0.55}

\begin{figure*}
\vspace{.2cm}
\label{fig:synthetic}
\begin{center}
\begin{tabular}{c c c}

\scalebox{\figscale}{
\begin{tikzpicture}
	\begin{axis}[
		xlabel={Fraction of outliers},
		ylabel={accuracy},
		xmax=1.15,
	]
	% use TeX as calculator:
	\addplot+[red,solid,mark=*,every mark/.append style={solid, fill=red}] table[x=p,y=mean_cs] {full_K50_N20.txt}; \addlegendentry{\small CS}
	\addplot+[blue,solid,mark=square,every mark/.append style={solid, fill=blue}] table[x=p,y=mean_sp] {full_K50_N20.txt}; \addlegendentry{\small SP}
	%\addplot+[red,dashed,mark=*,every mark/.append style={solid, fill=red}] table[x=p,y=mean_cs] {data/half_K50_N20.txt}; \addlegendentry{\small hConsensus}
	%\addplot+[black,solid,mark=triangle, thick] table[x=p,y=mean_in] {data/synth_full_K20_n20.txt};  \addlegendentry{\small Initial}
	\end{axis}
\end{tikzpicture}
}
&
%%%%%%%%%%
\scalebox{\figscale}{
\begin{tikzpicture}
	\begin{axis}[
		xlabel={Fraction of outliers},
		ylabel={accuracy},
		xmax=1.15,
	]
	% use TeX as calculator:
	\addplot+[red,solid,mark=*,every mark/.append style={solid, fill=red}] table[x=p,y=mean_cs] {full_K50_N50.txt}; \addlegendentry{\small CS}
	\addplot+[blue,solid,mark=square,every mark/.append style={solid, fill=blue}] table[x=p,y=mean_sp] {full_K50_N50.txt}; \addlegendentry{\small SP}
	%\addplot+[black,solid,mark=triangle, thick] table[x=p,y=mean_in] {data/synth_full_K20_n20.txt};  \addlegendentry{\small Initial}
	\end{axis}
\end{tikzpicture}
}
%%%%%%%%%%
&
\scalebox{\figscale}{
\begin{tikzpicture}
	\begin{axis}[
		xlabel={Fraction of outliers},
		ylabel={accuracy},
		xmax=1.15,
	]
	% use TeX as calculator:
	\addplot+[red,solid,mark=*,every mark/.append style={solid, fill=red}] table[x=p,y=mean_cs] {full_K50_N100.txt}; \addlegendentry{\small CS}
	\addplot+[blue,solid,mark=square,every mark/.append style={solid, fill=blue}] table[x=p,y=mean_sp] {full_K50_N100.txt}; \addlegendentry{\small SP}
	%\addplot+[black,solid,mark=triangle, thick] table[x=p,y=mean_in] {data/synth_full_K20_n20.txt};  \addlegendentry{\small Initial}
	\end{axis}
\end{tikzpicture}
}
%%%%%%%%%%%%%%%%%%%%%%%%%
\\
(a) $n=20$ & (b) $n=50$ & (c) $n=100$
\vspace{0.2cm}
\\
\scalebox{\figscale}{
\begin{tikzpicture}
	\begin{axis}[
		xlabel={Fraction of outliers},
		ylabel={accuracy},
		xmax=1.15,
	]
	% use TeX as calculator:
	\addplot+[red,solid,mark=*,every mark/.append style={solid, fill=red}] table[x=p,y=mean_cs] {half_K50_N20.txt}; \addlegendentry{\small CS}
	\addplot+[blue,solid,mark=square,every mark/.append style={solid, fill=blue}] table[x=p,y=mean_sp] {half_K50_N20.txt}; \addlegendentry{\small SP}
		
	%\addplot+[black,solid,mark=triangle, thick] table[x=p,y=mean_in] {data/synth_full_K20_n20.txt};  \addlegendentry{\small Initial}
	\end{axis}
\end{tikzpicture}
}
&
%%%%%%%%%%
\scalebox{\figscale}{
\begin{tikzpicture}
	\begin{axis}[
		xlabel={Fraction of outliers},
		ylabel={accuracy},
		xmax=1.15,
	]
	% use TeX as calculator:
	\addplot+[red,solid,mark=*,every mark/.append style={solid, fill=red}] table[x=p,y=mean_cs] {half_K50_N50.txt}; \addlegendentry{\small CS}
	\addplot+[blue,solid,mark=square,every mark/.append style={solid, fill=blue}] table[x=p,y=mean_sp] {half_K50_N50.txt}; \addlegendentry{\small SP}
	%\addplot+[black,solid,mark=triangle, thick] table[x=p,y=mean_in] {data/synth_full_K20_n20.txt};  \addlegendentry{\small Initial}
	\end{axis}
\end{tikzpicture}
}
%%%%%%%%%%
&
\scalebox{\figscale}{
\begin{tikzpicture}
	\begin{axis}[
		xlabel={Fraction of outliers},
		ylabel={accuracy},
		xmax=1.15,
	]
	% use TeX as calculator:
	\addplot+[red,solid,mark=*,every mark/.append style={solid, fill=red}] table[x=p,y=mean_cs] {half_K50_N100.txt}; \addlegendentry{\small CS}
	\addplot+[blue,solid,mark=square,every mark/.append style={solid, fill=blue}] table[x=p,y=mean_sp] {half_K50_N100.txt}; \addlegendentry{\small SP}
	%\addplot+[black,solid,mark=triangle, thick] table[x=p,y=mean_in] {data/synth_full_K20_n20.txt};  \addlegendentry{\small Initial}
	\end{axis}
\end{tikzpicture}
}
%%%%%%%%%%
\\
(d) $n=20$ &  (e) $n=50$ & (f) $n=100$
\end{tabular}
\caption{ Accuracy of consensus protocol (CS) and spectral method (SP) versus percentage of outliers in pairwise data associations.}
\end{center}
\end{figure*}
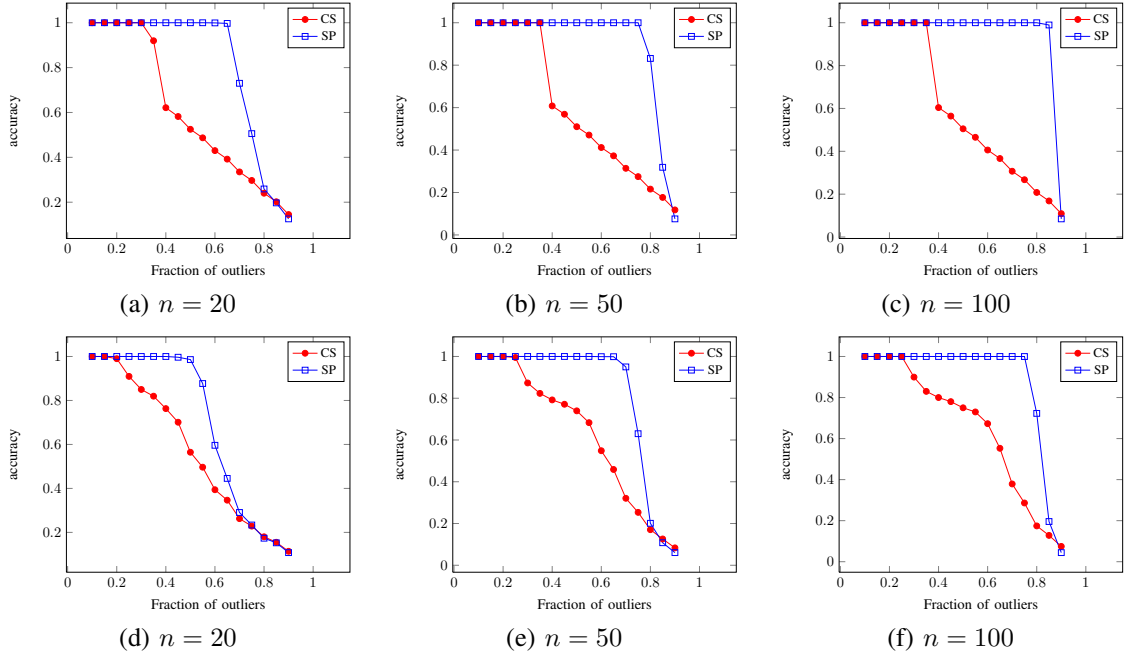

\subsection{Real data}

First, we experiment with the CMU  hotel sequence \footnote{ \url{http://vasc.ri.cmu.edu//idb/html/motion/hotel/}} which consists of 111 frames and the CMU house \footnote{\url{http://vasc.ri.cmu.edu//idb/html/motion/house}} sequence which consists of 101 frames.
We detect corner features \cite{harris1988combined} in the first image that remain visible across the entire sequence.
The ground-truth associations are obtained by tracking the features across the entire sequence.
Specifically, we used 89 points in the hotel sequence and 55 in the house sequence. 
Pairwise matches are computed between an image and at most 10 previous images and 10 subsequent images in the sequence, thus the sensor graph is relatively sparse. 
Pairwise matches are obtained by extracting SIFT descriptors \cite{lowe2004distinctive} and then solving the assignment problem using the Hungarian algorithm. We used the SIFT implementation of \cite{vedaldi2010vlfeat}.
The baseline method consists of obtaining labels by directly matching with the first image in the sequence.
The ground-truth associations are extracted by tracking the features over the sequence.
The results are presented in Table \ref{tab:cmudatasets}.
Both methods significantly outperform the baseline which shows that joint matching can significantly improve the overall matching performance.

\begin{table}[ht!]
\caption{CMU datasets Results} % title of Table
\label{tab:cmudatasets}
\centering % used for centering table
\begin{tabular}{|c|c|c|c|} % centered columns (4 columns)
%\hline
\hline 
Dataset & Baseline & Spectral & Consensus\\ [0.2ex] 
\hline  \hline
House &  0.716  &  0.861 &0.887  \\
Hotel&  0.666  &  0.872 & 0.878 \\
[0.2ex] % [1ex] adds vertical space
\hline %inserts single line
\end{tabular}
\end{table}

Next, we experiment with the Affine Covariant Regions Datasets \footnote{\url{http://www.robots.ox.ac.uk/~vgg/data/data-aff.html}} which consist of sequences of 6 images with significant overlap but with viewpoint variability.
We detect approximately 300 SIFT  features \cite{lowe2004distinctive} in each sequence that are visible across the entire sequence.
The ground-truth associations are extracted using the available ground-truth homographies. The results are presented in Table \ref{tab:graffitidataset}.
Both methods significantly outperform the baseline in the two challenging sequences Graffiti and Wall. 
However, in the remaining sequences the input pairwise associations are already very accurate resulting in only minor improvements if any.
It is worth noting that the spectral method and the consensus method have similar performance in all real sequences despite the fact that the spectral method
is computationally more expensive.
Therefore, in situations where the computational resources are limited, 
for instance swarms of UAVs with on-board cameras, the consensus algorithm would be preferable.  

\begin{table}[ht!]
\caption{Graffiti datasets Results} % title of Table
\label{tab:graffitidataset}
\centering % used for centering table
\begin{tabular}{|c|c|c|c|} % centered columns (4 columns)
\hline %inserts double horizontal lines
Dataset & Baseline & Spectral & Consensus\\ [0.5ex] % inserts table
%heading
\hline \hline 
Graffiti & 0.544   & 0.639 & 0.647  \\ 
Wall     & 0.570   & 0.646 & 0.658  \\
Trees    & 0.912   & 0.902 & 0.909  \\
Leuven   & 0.962   & 0.972 & 0.979 \\
Bikes    & 0.986   & 0.987 & 0.987  \\
UBC      & 0.996   & 0.996 & 0.996  \\
[0.2ex] % [1ex] adds vertical space
\hline %inserts single line
\end{tabular}
\label{table:nonlin} % is used to refer this table in the text
\end{table}

 \begin{figure*}
 \label{fig:realexamples}
 \vspace{.2cm}
 \centering
 \begin{tabular}{ccc}

\includegraphics[trim={6cm 13cm 6cm 13cm},clip,scale=1,width=5.2cm,height=3cm]{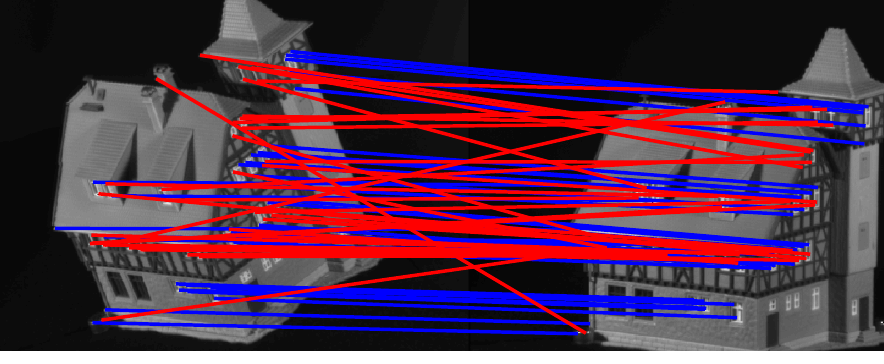}
%  trim={<left> <lower> <right> <upper>}
&
\includegraphics[trim={6cm 12.2cm 6cm 12.2cm},clip,scale=1,width=5.2cm,height=3cm]{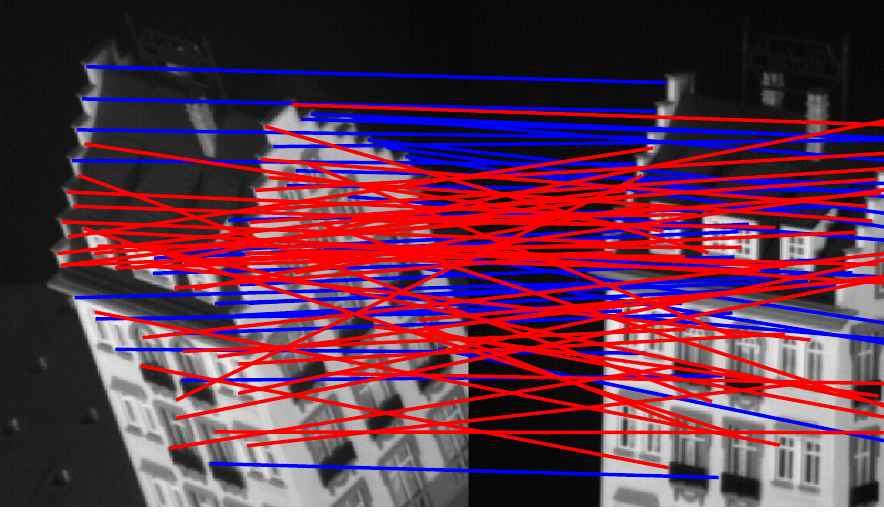}
&
\includegraphics[trim={6cm 12.6cm 6cm 12.6cm},clip,scale=1,width=5.2cm,height=3cm]{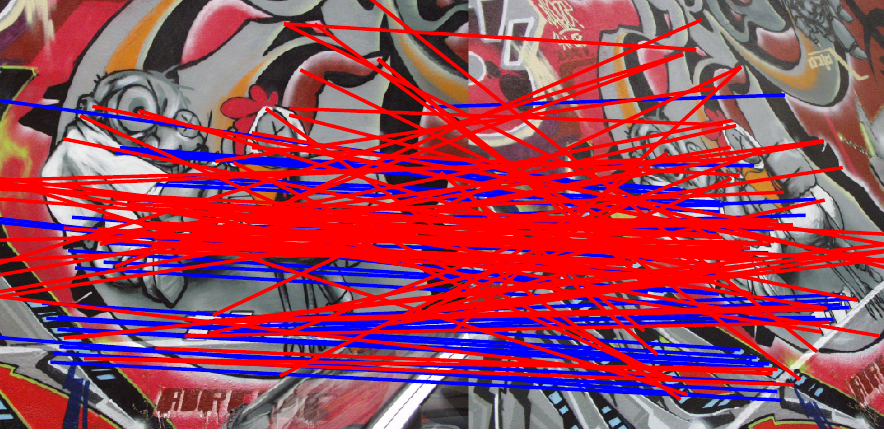}
\\
\includegraphics[trim={6cm 13cm 6cm 13cm},clip,scale=1,width=5.2cm,height=3cm]{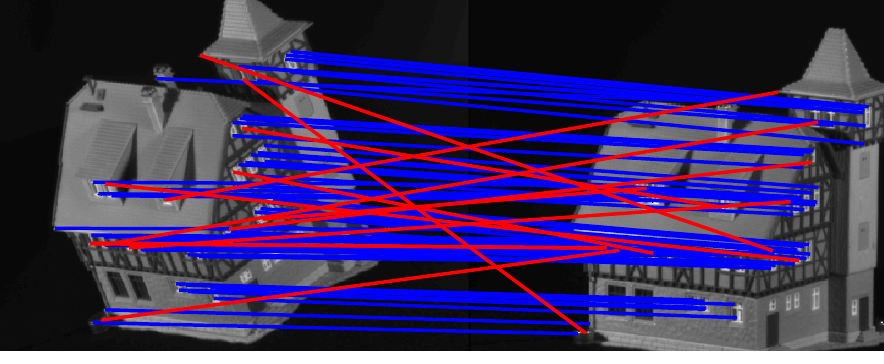}
&
\includegraphics[trim={6cm 12.2cm 6cm 12.2cm},clip,scale=1,width=5.2cm,height=3cm]{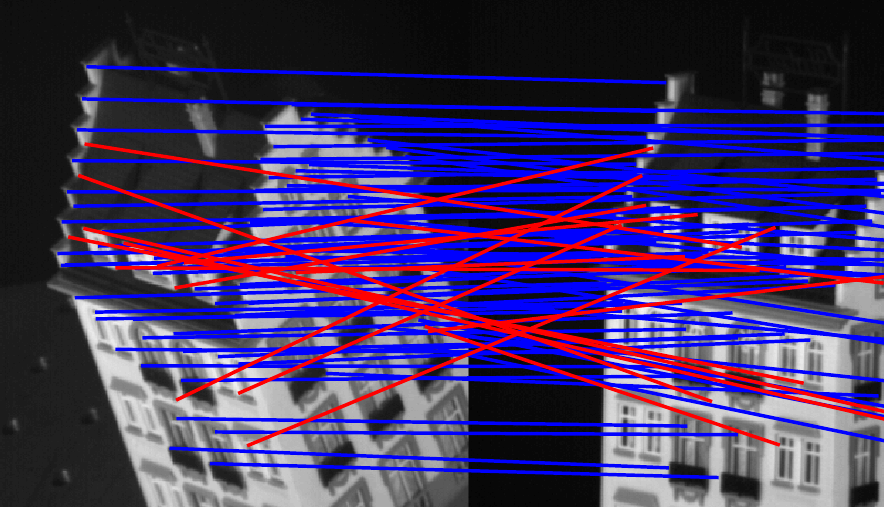}
&
\includegraphics[trim={6cm 12.6cm 6cm 12.6cm},clip,scale=1,width=5.2cm,height=3cm]{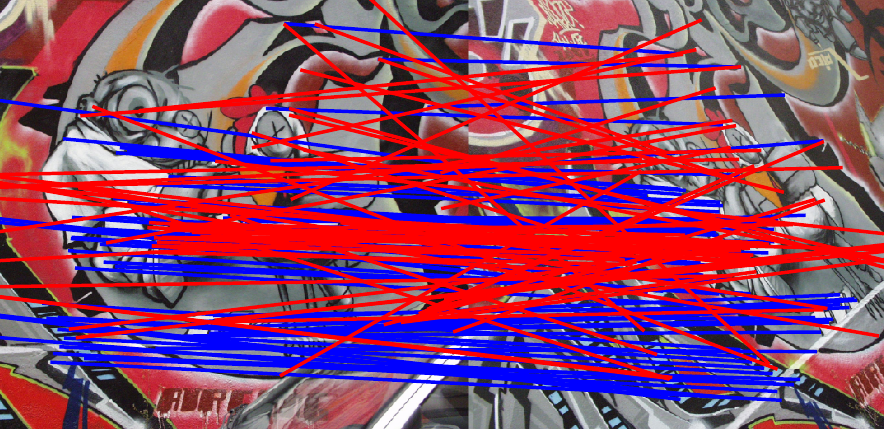}
 \end{tabular}
 \caption{Top: initial matches. Blue color corresponds to inlier matches whereas red color corresponds to outlier matches. Bottom: matches after the consensus algorithm.}
\end{figure*}

%%%%%%%%%%%%%%%%%%%%%%%%%%%%%%%%%%%%%%%%%%%%%%%%%%%%%%%%%%%%%%%%%%%%%%%%%%%%%%%%

\section{CONCLUSIONS AND FUTURE WORK}

In this paper, we proposed two fully decentralized methods for the problem of data association in sensor networks.
The first was a computationally inexpensive consensus-like algorithm and the second a decentralized spectral method.
We presented experimental results in the context of camera sensor networks. 
We believe that the results of this work will find applications in other settings as well.
In the future, we plan on generalizing these algorithms in order to handle different number of targets in each sensor and
time-varying pairwise associations.

%%%%%%%%%%%%%%%%%%%%%%%%%%%%%%%%%%%%%%%%%%%%%%%%%%%%%%%%%%%%%%%%%%%%%%%%%%%%%%%%
%\addtolength{\textheight}{-12cm}
%%%%%%%%%%%%%%%%%%%%%%%%%%%%%%%%%%%%%%%%%%%%%%%%%%%%%%%%%%%%%%%%%%%%%%%%%%%%%%%%

\section{APPENDIX}


\subsection{Proof of lemma \ref{lem:perronspecial}}
\label{app:proofoflemmaperronspecial}

First, we need the following lemma:
 \begin{lemma}
\label{lem:detzero}
Let $L_i(\graph)$ be the matrix obtained from the Laplacian $L(\graph)$ by removing the $i$th row and $i$th column. Define $F_i(\graph)$ analogously. Then,
\begin{equation}
\det (L_i(\graph)) \neq 0 \quad \textrm{iff} \quad  \det (I -F_i(\graph)) \neq 0
\end{equation}
\end{lemma}

\medskip
\noindent
\begin{proof}
By construction $L(\graph) = (I+\Delta(\graph)) ( I- F(\graph))$ and thus, $L_i(\graph) = (I+\Delta_i(\graph)) ( I- F_i(\graph))$  since $\Delta(\graph)$ is diagonal.
Therefore, since   $\det (I+\Delta_i(\graph)) >0$, we have that $\det(L_i(\graph)) \neq 0$ if and  only if $\det (I-F_i(\graph)) \neq 0$.
\end{proof}

Now, we are ready to proceed to the proof of lemma \ref{lem:perronspecial}
By Gers\v gorin's discs theorem, we know that the eigenvalues of $F(\graph)$ lie in the unit circle and since it is stochastic, that is $F\mathbf{1}=\mathbf{1}$, it can be immediately concluded that $\rho(F(\graph)) = 1$.

From  Proposition 3.8 page 51 of \cite{mesbahi2010graph}, we have that if the (directed) graph $\graph$ contains a rooted-out branching as a subgraph, then $\rank L(\graph) = n-1$. Combining this result with \ref{lem:detzero}, we conclude that the existence of a rooted-out branching implies that $\rank (I-F(\graph)) = n-1$. We conclude that $1$ is a simple eigenvalue of $F(\graph)$.

Let $\lambda_1,\lambda_2,\ldots,\lambda_n$ be an enumeration of the eigenvalues of $F(\graph)$  such that $\lambda_1=1$. Moreover, let $F(\graph) = P J(\Lambda) P^{-1}$ be the Jordan decomposition of $F(\graph)$. From,  $F(\graph) P = P J(\Lambda)$, it is easy to see that the the first column of $P$ is in the span of $ \mathbf{1}$ and thus, $\lim_{k \rightarrow \infty} F(\graph)^k = \mathbf{1} c^T $. Moreover, since $F(\graph)^k$ is stochastic for all positive integers $k$,  it follows that $c \geq 0$ and $c^T \mathbf{1} = 1$.

%\subsection{A useful lemma}

%\begin{lemma}
%\label{lem:rhoconvergence}
%Assume that  $\rho(A^k) < 1$ for some integer $k >0$. Then, $\lim_{n\rightarrow \infty} A^n = 0$.
%\end{lemma}

%\label{app:proofofrhoconvergence}
%For $k=1$, it is well known that $\rho(A)<1$ implies $\lim_{n\rightarrow \infty} A^n = 0$. Now, assume it holds for some $k >1$. Since $\rho(A^k)<1$, it follows that the subsequence $A^{km}$ for $m=1,2,\ldots$ converges to $0$. Let $\norm{\cdot}$ denote any matrix norm and let $\mu = \norm{A}$. For any integer $l$ $ 0 \leq l  \leq k$, we have
%\[\norm{A^{km+l}} \leq  \norm{A^{km}} \mu^l \leq \norm{A^{km}} \mu^k\]
%For any $\epsilon >0$, pick $m_0$ large enough such that $\norm{A^{km_0}} < \epsilon/\mu^k $. Such a $m_0$ is guaranteed to exist by the convergence of the subsequence $A^k,A^{2k},\ldots $. Then,  $\norm{A^{km_0+l}}  \leq  \epsilon$
%which shows that all subsequences of the form $A^{km+l}$ converge and therefore, $\lim_{n\rightarrow \infty} A^n$ converges.
%\QEDA

 \subsection{Proof of theorem \ref{thm:covergence_noiseless}}
\label{app:app_covergence_noiseless}

%We will use the fact that for any square matrix $B$ and any $k \geq 1$, we have $(B\otimes I)^k = B^k \otimes I$ .

In the noiseless case, observe that
\begin{equation}
F(\cD) = \mathbf{\Pi}_0 \left( F(\graph) \otimes I_m \right) \mathbf{\Pi}_0^T,
\end{equation}
where  $ \mathbf{\Pi}_0 = \diag (\Pi_{10},\ldots, \Pi_{n0})$, each $\Pi_{i0} \in \symgroup{m}$  satisfies $\Pi_{i0} \Pi_{j0}^T=\tPi_{ij}$.
\begin{equation}
\lim_{k \rightarrow \infty} F(\cD)^k = \mathbf{\Pi}_0  \left(  F(\graph)^\infty \otimes I_m \right)\mathbf{\Pi}_0 ^T
\end{equation}
since $\lim_{k\rightarrow \infty} F(\graph)^k $ exists as we proved earlier.
It follows that $F(\cD)^k$ converges to a limit as well.
Any asymptotic solution of the consensus protocol satisfies $ \mathbf{\Pi}(\infty) = F(\cD) \mathbf{\Pi}(\infty)$.
  Then, for all $i \in \{1,\ldots,n\}$ we have $\Pi_i(\infty) = ({1}/{\card{\cNi}}) \sum_{j \in \cNi} \tPiij \Pij(\infty)$.  Using $\tPiij =  \Pi_{i0} \Pi_{j0}^T$, we get
  \[
( L(\graph) \otimes I_m )
\begin{bmatrix}
\Pi_{10}^T \Pi_1(\infty) \\ \vdots \\ \Pi_{n0}^T \Pi_n(\infty)
\end{bmatrix}
=0
\]
Under the assumption that $\graph$ is contains a rooted-out branching and thus, $\textrm{null} ( L(\graph)) = \textrm{span} (\mathbf{1})$,  we obtain $\Pi_i(\infty) = \Pi_{i0} Q$ where $Q$ is an arbitrary matrix $Q$. Moreover, if one enforces $\Pi_1(t) = \Pi_1(0)$, $\Pi_1(0) \in \symgroup{m}$, for all $t \geq 0$, it follows that $Q=\Pi_{10}^T \Pi_1(0$ and $\Pi_i(\infty) =  \Pi_{i0} \Pi_{10}^T \Pi_1(0)$, that is equal to any valid labels up to a global permutation.
\QEDA

\subsection{Proof of lemma \ref{lem:perronspecial2}}

\label{app:proofperronspeical2}

Part (a) follows from the fact that $F(\graph)$ is stochastic.
Without loss of generality, assume that vertices $\{1,2,\ldots,m\}$ are the distinguished ones. Then, $F(\graph)$ takes the form
\begin{equation}
F(\graph) =
\begin{bmatrix}
I_m & 0 \\
F_{21} & F_{22}
\end{bmatrix}
\end{equation}
For all $k=1,2,\ldots$ we have by a simple induction:
\begin{equation}
F(\graph)^k =
\begin{bmatrix}
I_m & 0 \\
F_{21,k} & F_{22}^k
\end{bmatrix}
\end{equation}
where $F_{21,1} = F_{21}$, $F_{21,k+1} = F_{21,k}+F_{22}^k F_{21}$

By the assumption that every non-distinguished vertex is path-connected to some distinguished vertex, it follows that for $k>0$ large enough, each row of $F_{21,k}$ contains at least one positive entry. Thus, all eigenvalues of $F_{22}^k$ lie strictly inside the unit circle and thus, $\lim_{k \rightarrow \infty} F_{22}^k = 0$.
\QEDA
%%%%%%%%%%%%%%%%%%%%%%%%%%%%%%%%%%%%%%%%%%%%%%%%%%%%%%%%%%%%%%%%%%%%%%%%%%%%%%%%

%\section*{ACKNOWLEDGMENT}

%\textcolor{red}{The preferred spelling of the word ÒacknowledgmentÓ in America is without an ÒeÓ after the ÒgÓ. Avoid the stilted expression, ÒOne of us (R. B. G.) thanks . . .Ó  Instead, try ÒR. B. G. thanksÓ. Put sponsor acknowledgments in the unnumbered footnote on the first page.}

%%%%%%%%%%%%%%%%%%%%%%%%%%%%%%%%%%%%%%%%%%%%%%%%%%%%%%%%%%%%%%%%%%%%%%%%%%%%%%%%

\bibliographystyle{plain}
\bibliography{egbib}

\end{document}